\documentclass[ ]{amsart}
\usepackage[utf8]{inputenc}

\usepackage{color}
\usepackage{mathrsfs}
\usepackage{graphicx}

\usepackage[shortlabels]{enumitem}
\setlist[enumerate, 1]{i)}
\usepackage[T1]{fontenc}
\usepackage{pifont}
\usepackage{amsmath}
\usepackage{amssymb}
\usepackage{amsthm}
\usepackage{amscd}
\usepackage{bm}

\newtheorem{theorem}{Theorem}
\newtheorem{lemma}[theorem]{Lemma}

\theoremstyle{remark}
\newtheorem{remark}{Remark}

\newcommand{\R}{\mathbb{R}}
\newcommand{\C}{\mathbb{C}}

\renewcommand{\Re}{\operatorname{Re}}

\newcommand{\id}{\,\mathrm{d}}

\newcommand{\mh}{\mathcal{H}}
\newcommand{\abs}[1]{\lvert#1\rvert}
\newcommand{\norm}[1]{\lVert#1\rVert}
\newcommand{\bvec}[1]{\boldsymbol{#1}}

\newcommand{\dotprod}[2]{
  \bvec{#1}\mkern1mu{\cdot}\mkern1mu\bvec{#2} \,}
  
\newcommand{\inprodS}[2]{\left ( #1 , #2\right )_{\C^2}}  
\newcommand{\Scalprod}[2]{\left \langle #1 , #2\right \rangle }  

\newcommand{\Dom}[1]{\mathcal{D}( #1)}

\newcommand{\D}{D}
\newcommand{\T}{T}
\newcommand{\Trace}{\gamma}

\newcommand{\Bound}{{\partial \Omega}}
\newcommand{\curv}{\kappa}

\title[Spectral gaps of Dirac operators]{Spectral gaps of Dirac operators describing graphene quantum dots}
\author[Benguria]{Rafael D. Benguria}
\author[Fournais]{S\o ren Fournais}
\author[Stockmeyer]{Edgardo Stockmeyer}
\author[Van Den Bosch]{Hanne Van Den Bosch}
\address{ Rafael D. Benguria, Edgardo Stockmeyer and Hanne Van Den Bosch\\ Instituto de F\'\i sica\\
Pontificia Universidad Cat\'olica de Chile\\
Vicu\~na Mackenna 4860\\
 Santiago 7820436, Chile.}
\address{S\o ren Fournais, Department of Mathematics, Aarhus University, Ny Munkegade 118, DK-8000 Aarhus, Denmark}
\begin{document}

\begin{abstract}
  The two-dimensional Dirac operator describes low-energy excitations
  in graphene. Different choices for the boundary conditions give rise to qualitative differences in the spectrum of the resulting operator. For a family of boundary conditions, we find a lower bound to the spectral gap around zero, proportional to
  $\abs{\Omega}^{-1/2}$, where $\Omega \subset \R^2$ is the bounded region where the Dirac operator acts. This family contains the so-called infinite mass and armchair cases used in the physics literature for the description of graphene quantum dots. 
\end{abstract}

\maketitle

\section{Introduction}
Graphene is a two-dimensional layer of carbon atoms forming a
honeycomb lattice.  Due to its hexagonal symmetry
\cite{wallace,Fefferman}, in absence of external fields, low-energy electronic excitations in an extended graphene sheet behave as Dirac fermions. Their dynamics is described
effectively by the Hamiltonian
\begin{align}\label{h}
  H=
\begin{pmatrix}
                                        T &0\\0&T
                                       \end{pmatrix}\quad\mbox{on}
\quad \mh\oplus \mh,
\end{align}
where $\mh=L^2(\R^2,\C^2)$ and $T$ is the massless two-dimensional
Dirac operator
\[
\T = v_f \hbar( -i \dotprod{\sigma}{}\nabla) \quad\mbox{on}
\quad \mh,
\]
where $v_f\sim 10^6 \, \rm{m}/\rm{s}$ is the Fermi velocity. Here $\bvec
\sigma = (\sigma_1, \sigma_2)$ are the first two Pauli
matrices. Through the orthogonal sum in \eqref{h}, the operator $H$
takes into account contributions from the two inequivalent Dirac
points (or valleys) $\bvec K$ and $\bvec K'$ of the first Brillouin
zone associated to the lattice.  The components of a wavefunction in
$\mh$ describe the electronic density on each of the two triangular
sublattices that constitute the honeycomb lattice.  In many
applications the contributions from the two valleys do not couple and
the description is reduced to the study of the operator $T$ only (see
\cite{CastroNetoetAl} for a review). 

When considering (quasi-)particles confined to some region $\Omega\subset
\R^2$, one should impose boundary conditions that might break the block-diagonal structure of $H$. In the physics literature, much attention has been devoted to the so-called zigzag, armchair and infinite mass boundary conditions.
The choice of a boundary condition influences the spectrum and therefore the transport properties of graphene ribbons and flakes, see for instance \cite{AkhmerovBeenakker, mccannfalko04, WurmRycerzetAl} for theoretical considerations or \cite{RitterLyding} for experimental observations. In particular, the presence of a gap in its spectrum allows to use a graphene device as a semiconductor.

We move on to a brief description of the above mentioned boundary conditions.
Zigzag and armchair boundary conditions emerge from the tight-binding model
and correspond to two different orientations of a straight lattice termination
\cite{AkhmerovBeenakker, mccannfalko04}.  The zigzag boundary
conditions are known to be gapless, having zero as an eigenvalue of infinite multiplicity. From the
mathematical point of view this has been observed in
\cite{schmidt1994} (see also \cite{freitassiegl2014} for the absence
of gaps of certain perturbed zigzag operators). The
associated zero-energy states are well localized close to the
boundary. In contrast, for armchair boundary conditions the presence of a gap has been noted (see
 e.g. \cite{PhysRevB.73.235411,PhysRevB.88.125409, ZhengWangetAl}), and the lowest energy states are rather delocalized. 
Infinite mass boundary conditions, on the other hand, do not arise from the lattice termination.
In fact, they were first studied in 1987 by Berry and Mondragon \cite{BerryMondragon} for the operator $T$. They emerge from the Dirac operator with an effective mass term supported outside $\Omega$, as a limiting case when the mass tends to infinity, see \cite{BerryMondragon, stockmeyerwugalter2016}. In the description of graphene, infinite mass boundary conditions have been also used to model quantum dots or nano-ribbons exhibiting a gap independent of the lattice orientation
\cite{AkhmerovBeenakker,beneventano2014charge, CastroNetoetAl, PonomarenkoetAl}. 
Zigzag and infinite mass boundary conditions do not couple the valleys and thus can be defined for $T$ too. Armchair boundary conditions do mix the valleys and make sense only when considering the full operator $H$.

In the present work, we obtain a lower bound on the gap size for $H$ with armchair or infinite mass boundary conditions in terms of $|\Omega|$. We start our analysis by studying the Dirac operator $T$ on a bounded simply connected domain $\Omega \subset \R^2$. In Theorem~\ref{thm : lower_bound}, we prove the desired estimate for a certain class of boundary conditions including the infinite mass case. The proof of the theorem will be given in Section~\ref{sec : lower_bound}. In Section~\ref{sec : 2-valley}, through an elementary observation, we show that the lower bound for the infinite mass operator $T$ applies equally well to the operator $H$ with armchair boundary conditions. In this section, we also provide some further details on the physically relevant boundary conditions. We complete this introduction with the necessary definitions and the precise statement of the theorem. 
In our proof we follow the 
scheme developed by B\"ar in \cite{Bar1992}. However in his case, B\"ar considers a manifold with curvature but  without boundary. Towards the end of our proof we need to choose 
a trial function $f$. This is similar to what is done in \cite{Bar1992}. In B\"ar's case his choice is dictated by the curvature of the manifold while in our case the choice is related 
to a boundary value problem which depends on our boundary conditions (see equation (\ref{eq : neumann})).

\bigskip
\noindent
{\it Note added in proof:}  When preparing this manuscript, we were not aware of the work of Raulot \cite{Raulot06} and we thank the referee for pointing out this reference to us. The case $n=2$ of \cite[Theorem 1]{Raulot06} is a generalization of our Theorem 1 with $B =1$, to arbitrary manifolds with boundary and so-called chiral boundary conditions. It is interesting to note that these chiral boundary conditions from differential geometry coincide for manifolds $\Omega \subset \R^2$ with the infinite mass boundary conditions known in the literature on graphene. Here, we stress the connection with the physics of graphene and provide a self-contained proof. It does not require sophisticated tools from differential geometry and also works in cases with limited regularity. Through an elementary observation (see Lemma 2) we can treat a family of boundary conditions ($0< B \le 1$).

\subsection{Definitions and the main theorem}
We consider a two-dimensional Dirac operator on a bounded domain $\Omega \subset \R^2$ 
with $C^2$-boundary $\partial \Omega$.
Choosing appropriate units, the Dirac operator acts as the differential expression
\[
T \equiv -i \dotprod{\sigma}{} \nabla  = \sigma_1 (-i \partial_1) 
                                       + 
                                        \sigma_2
                                      (-i \partial_2). 
\]
Here, the Pauli matrices are defined as
\[
\sigma_1 =\begin{pmatrix}
 0&1\\1&0
\end{pmatrix},
\quad
\sigma_2 = \begin{pmatrix}
 0&-i\\i&0
\end{pmatrix},
\quad
\sigma_3 = \begin{pmatrix}
 1&0\\0&-1
\end{pmatrix},
\]
and we use the convention $\dotprod{\sigma}{v} = \sum_{i=1}^{3} v_i \sigma_i $.
We will write $\D_\eta$ for the operator acting as $\T$
on functions in the domain 
\[
\Dom{\D_\eta} \equiv  \{ u \in H^1(\Omega, \C^2)| P_{-,\eta} \Trace u = 0\}.
\]
Here $\Trace$ is the trace operator on the boundary of $\Omega$ and the orthogonal projections $P_{\pm, \eta}$ are defined as
\[
P_{\pm,\eta} = 1/2(1 \pm                           A_\eta),
                                                                       \quad A_\eta = \cos(\eta) \dotprod{\sigma}{t}+ \sin(\eta) \sigma_3,
\]
where $\bvec t$ is  the unit vector tangent to the boundary.
This is the only family of local boundary conditions making $T$ into a symmetric operator on $H^1(\Omega)$.
A priori, $\eta$ can be any real function of $\Bound$, but in the physically relevant cases it is a constant on each connected component of $\Bound$. Infinite mass boundary conditions correspond to $\eta \equiv 0$ or $\eta \equiv \pi$.
If $\eta$ is $C^1$ and $\cos \eta (s) \neq 0$ for all $s \in \Bound$, $D_\eta$ is self-adjoint \cite{BenguriaFournaisetAl}. 
In this case, it follows from the compact embedding of $H^1(\Omega) $ in $ L^2(\Omega)$ that the resolvent of $D_\eta$ is compact.
Thus, its spectrum consists of eigenvalues of finite multiplicity accumulating only at $\pm \infty$.

For constant $\eta$ and simply connected domains,
we obtain the following lower bound for the spectral gap.
\begin{theorem} \label{thm : lower_bound}
 Take $\Omega \subset \R^2$ simply connected with $C^2$-boundary. Let $\eta$ be a constant such that $\cos \eta \neq 0$ and define $\D_\eta$ as before.
 Define $B = \min (\abs{\cos\eta/(1- \sin\eta)}, \abs{(1- \sin\eta)/\cos \eta)}$.
 If $\lambda$ is an eigenvalue of $\D_\eta$, then
 \[
\lambda^2 \geq \frac{2 \pi}{\abs{\Omega}} B^2.
\]
\end{theorem}
\begin{remark}
In physical units our lower bound for infinite mass
boundary conditions ($B = 1$ in the theorem) gives a gap larger
than $2\sqrt{2\pi} \hbar v_{f} |\Omega|^{-1/2} $. This means that in
order to obtain a gap of $1 \,\mathrm{e} \rm {V}$ one needs a domain with a diameter of about $10
\,\rm{nm}$.
\end{remark}
\begin{remark}
The bound is not sharp, but it is quite good, as a comparison with the case of a disc shows.
When $B = 1$, the lowest eigenvalue for a disc of unit radius is $k_0$, the smallest positive number such that $J_0(k_0) = J_1(k_0)$,
where $J_n$ is the $n$-th Bessel function of the first kind (see \cite{BerryMondragon}). 
Numerically, $k_0 \approx 1.435$, and our lower bound reads
\[
k_0 > \sqrt{2 } \approx 1.414.
\]
It is an open problem to obtain a sharp bound among all bounded two dimensional domains with the type of boundary conditions we consider. 
\end{remark}
\begin{remark}
 The bound obtained by Raulot \cite{Raulot06} is sharp and the case of equality is obtained by a manifold isomorphic to a half sphere
$\mathbb{S}_+^n(r)$ with radius $r$, where $r$ depends on the first eigenvalue of the Dirac operator on the manifold with the boundary conditions he considers. 
\end{remark}

In the differential geometry literature, much attention has been
devoted to lower bounds for the square of Dirac operators on surfaces.
Most of these results deal with closed surfaces
\cite{Friedrich1980,Bar1992}. For Dirac operators on two-dimensional
manifolds with boundaries a less explicit bound has been derived in
\cite{HijaziMontielZhang}.  Our proof uses ideas from \cite{Bar1992}.

\subsection{Notation}\label{notation}
Before going further, we need to fix some notations.
We will consider a fixed domain $\Omega$ with $C^2$-boundary $\Bound$.
We denote by $\bvec{n}(s)$ and $\bvec{t}(s)$ the outward normal and the tangent vector to the boundary at the point $s \in \Bound$.
The orientation of $\bvec t$ is chosen such that $\bvec n, \bvec t$ is positively oriented, so we have
$\dotprod{t}{} \nabla  \bvec t(s):=\partial_s \bvec t (s) = - \curv(s) \bvec n (s)$, where $\curv (s)$ is the curvature of the boundary.
If $\bvec t (s) = (t_1(s), t_2(s))$, we define $t(s) = t_1(s)+ i t_2 (s)$, the tangent vector seen as a number in $\C$.

Functions in $H^1(\Omega)$ have well-defined traces on $\Bound$, and since this will not cause confusion, we use the same notation
for a function and its trace.
In $L^2 (\Omega, \C^2)$, the notations $\Scalprod{\cdot}{\cdot}$ and $\norm{\cdot}$ will be used for the inner product and norm, respectively.

\section{Proof of the theorem} \label{sec : lower_bound}
Before going into the proof of Theorem~\ref{thm : lower_bound}, let us give a heuristic interpretation of the result.
In \cite{Bar1992} the following lower bound for the eigenvalues of the (classical) Dirac operator
on \emph{closed} surfaces $M$ of genus one (surfaces homeomorphic to a sphere) is proved:
\begin{equation} \label{eq : Bar}
 \lambda^2 \geq \frac{4 \pi}{\mathrm{area}(M)}.
\end{equation}

The bound we obtain for an open and simply connected surface $\Omega$
with boundary condition $\eta \equiv 0$, $B= 1$, is

 \[
\lambda^2 \geq \frac{2 \pi}{\abs{\Omega}}.
\]

At least formally, these particular boundary conditions provide the possibility to extend spinors in $\Dom{\D}$
to the \emph{invertible double} $\widetilde \Omega$, which is the closed surface obtained by glueing $\Omega$ to its mirror image.
Details of this construction can be found in \cite[chapter 9]{bookBooss}.
The bottom line is that an eigenspinor $u $ of $\D$ can be extended to an eigenspinor $\widetilde u$ of the extended Dirac operator $\widetilde\D$
by identifying $\widetilde u \approx (u, - u)$. 
Then, the bound of Theorem \ref{thm : lower_bound} follows from B\"ar's bound \eqref{eq : Bar}, since $\mathrm{area}(\widetilde \Omega) = 2 \abs \Omega$.

This argument can be made rigorous by considering closed surfaces $\widetilde \Omega_\epsilon \subset \R^3$ consisting of two copies of $\Omega$ 
in parallel planes with a distance $\epsilon$ between them, joined smoothly by a \emph{ribbon} of width proportional to $\epsilon$.
There is some work involved in computing explicitly the extension of $\D$ 
to the curved ribbon and in checking that eigenspinors can be extended correspondingly.
One has to make sure that the contribution of the curved part to the
Rayleigh quotient tends to zero with $\epsilon$ in order to obtain the
result. If $\Omega$ is not simply connected, we can still perform the doubling
construction, but the resulting closed manifold will be homeomorphic
to a torus or a surface of higher genus.  In principle, this case can
be treated using the results in \cite{AmmannBar} that extend \eqref{eq
: Bar}.

Instead of going through calculations with spinors on curved surfaces
we will use the strategy from \cite{Bar1992} taking care of the
boundary terms. The boundary conditions for constant $\eta \notin \{
0, \pi \}$ do not have the above described {\it doubling property}. However, one can
extend the result for $\eta = 0$ to the general case, as the following
lemma shows.
\begin{lemma}
Take $\D_\eta$ satisfying the hypotheses of Theorem~\ref{thm : lower_bound}, 
and 
$$
B = \min \bigl(\abs{\cos\eta/(1- \sin\eta)}, \abs{(1- \sin\eta)/\cos \eta} \bigr).
$$
 If $\lambda_\eta$ is the eigenvalue of $\D_\eta$ of smallest absolute value, then
 \[
 \lambda_\eta^2 \geq B^2 \lambda_0^2.
 \]
\end{lemma}
\begin{proof}
 Assume $\eta \in (0, \pi/2)$ such that $B = (1- \sin\eta)/\cos\eta \in
 (0,1)$.
 Take  an eigenspinor $u$ of $\D_\eta$ associated to the eigenvalue $\lambda_\eta$.
 Writing out the boundary conditions explicitly, we obtain 
 $u_2 = B t u_1 $ on $ \Bound$.
 Then, we may write $u = v + w$,
 where $v =\left( \begin{smallmatrix}
               B & 0 \\ 0 & 1 
              \end{smallmatrix} \right) u$.              
 This gives $v \in \Dom{\D_0}$, while $w \in \Dom{\D_{\pi/2}}$.
 Now we have
 \[ 
 \lambda_\eta^2 \norm{u}^2 = \norm{\T v + \T w}^2 = \norm{\T v}^2 + \norm{\T w}^2 + 2 \Re \Scalprod{\T v}{\T w}.
  \]
 The last two terms can be combined using the fact that $w$ has only
 its first component nonzero. We get
 \begin{align*}
 \norm{\T w}^2 + & 2 \Re \Scalprod{\T v}{\T w} \\
 & = (1-B)^2 \norm{(-i\partial_1 - \partial_2) u_1}^2 
 + 2 B (1-B) \norm{(-i\partial_1 - \partial_2) u_1}^2\\ & = (1-B^2)\norm{(-i\partial_1 - \partial_2) u_1}^2.
 \end{align*}
 Since $\abs B \leq 1 $ by definition, 
 we have
  \[ 
 \lambda_\eta^2 \norm{u}^2 \geq \norm{\D_0 v}^2 \geq \lambda_0^2 \norm{v}^2 \geq \lambda_0^2 B^2 \norm{u}^2,
  \]
  which is the desired inequality.  
  The other cases are analogous: it suffices to define $v = \left( \begin{smallmatrix}
              - B & 0 \\ 0 & 1 
              \end{smallmatrix} \right) u$  when $\eta$ lies in
            $(\pi/2,\pi)$ or $v = \left( \begin{smallmatrix}
                1 & 0 \\ 0 & \pm B 
             \end{smallmatrix} \right) u$ for $\eta\in (\pi,3\pi/2)$
           and $\eta\in(3\pi/2,2\pi)$.
\end{proof}

\begin{proof}[Proof of Theorem \ref{thm : lower_bound}]
By the previous lemma we can restrict our attention to $\eta = 0$, so to simplify notations, we will write $\D_0 = \D$.
Recall that the Pauli matrices satisfy the (anti-)commutation relations
 \[
 \{\sigma_j, \sigma_k \} = 2\delta_{jk}, \quad 
 [\sigma_j, \sigma_k]  =  2 i \epsilon_{jkl} \sigma_l, \quad j,k,l \in \{ 1,2,3 \},
 \]
where $\delta_{jk}$ is the Kronecker delta and $\epsilon_{jkl}$ is the Levi-Civita symbol, 
which is totally antisymmetric and normalized by $\epsilon_{123} = 1$.
We start by a calculation for $C^1$-spinors $u,v \in \Dom{\D}$
\begin{align*}
(\D u, \D v)
	  &= \sum_{k,j}\int_\Omega \inprodS{\partial_k u}{ \sigma_k \sigma_j \partial_j v} \\
	 & = \sum_{k} \int_\Omega \inprodS{\partial_k u}{\partial_k v} + i \sum_{k,j} \epsilon_{kj3} \int_\Omega \inprodS{\partial_k u}{\sigma_3 \partial_j v}.
\end{align*}
In the second term we can integrate by parts using the antisymmetry of $\epsilon_{kj3}$ 
and introduce the tangent vector at the boundary $\bvec t = (- n_2,
n_1)$. We obtain
\begin{align*}
\sum_{k,j} i \epsilon_{kj3} \int_\Omega \inprodS{\partial_k u}{\sigma_3 \partial_j v} 
	& = \sum_{k,j} i \epsilon_{kj3} \int_\Omega \partial_k \inprodS{u}{\sigma_3 \partial_j v} \\
	& =\sum_{k,j}i \epsilon_{kj3} \int_\Bound \bvec{n}_k \inprodS{ u}{\sigma_3 \partial_j v} \\
	& = i \int_\Bound \inprodS{u}{\sigma_3 \dotprod{t}{}\nabla v}.
\end{align*}
Since only the tangent derivative is involved, this term depends solely on the boundary values of $u$ and $v$.
We can explicitly write out the spinor components and introduce the boundary condition in the form $u_2 = t u_1 $ :
\begin{align*}
 \inprodS{u}{\sigma_3 \dotprod{t}{}\nabla v}
  = u_1^* \dotprod{t}{}\nabla v_1 - u_2^*  \dotprod{t}{}\nabla v_2 = - u_1^* v_1 t^* t' = - i u_1^* v_1 \curv(s) 
\end{align*}
In the last equality, we used $\partial_s \bvec t (s) = -\curv(s)
\bvec n (s)$ (see Subsection \ref{notation}).
By density,
\begin{equation} \label{eq : qform}
 (\D u, \D v) = ( \nabla u,  \nabla v) + \frac{1}{2}\int_\Bound \inprodS{u}{v}(s) \curv(s) \id s 
\end{equation}
holds for all $u , v \in \Dom{\D}$.

For a real constant $\alpha$ and a real $C^1$-function $f$
we define a modified connection
\[
\widetilde \partial_j = \partial_j -i \alpha \sigma_j  - \sigma_f \sigma_j , 
\]
where $\sigma_f := \dotprod{\sigma}{}\nabla f$.
For spinor fields $u, v$ we compute the product
\begin{align*}
 \sum_j \inprodS{\widetilde \partial_j u}{\widetilde \partial_j v}
	    =& \sum_j \Bigl( \inprodS{\partial_j u}{ \partial_j v} 
	    +  \alpha^2 \inprodS{\sigma_ju}{\sigma_jv} +  \inprodS{\sigma_f \sigma_j u}{\sigma_f \sigma_j v}
	    \\
	    &\quad - \inprodS{\partial_j u}{(i \alpha \sigma_j +  \sigma_f \sigma_j   )v}
		- \inprodS{ (i \alpha \sigma_j +  \sigma_f \sigma_j  ) u}{\partial_j v} \\
	    &\quad  +i\alpha  \inprodS{\sigma_j u}{ \sigma_f \sigma_j v}
		-i  \alpha \inprodS{\sigma_f \sigma_j u}{\sigma_j v}  \Bigr)  \\
	    =& \inprodS{ \nabla u}{  \nabla v}  + (2 \alpha^2 + 2 \abs{\nabla f}^2) \inprodS{u}{v} \\
		&\quad - \alpha (\inprodS{Du}{v} + \inprodS{u}{Dv}) \\
		& \quad - \sum_j \bigl(\inprodS{\partial_ju}{\sigma_f \sigma_j v} + \inprodS{ \sigma_f \sigma_j u}{\partial_j v} \bigr)	    .
\end{align*}
By the anti-commutation relations, we obtain
\[
\sum_j \inprodS{ \sigma_f \sigma_j u}{\partial_j v} = - i \inprodS{ u}{\sigma_f D v} + 2 \sum_j (\partial_j f) \inprodS{u}{\partial_j v},
\]
so 
\begin{align*}
 \sum_j \inprodS{\widetilde \partial_j u}{\widetilde \partial_j v}
	       =& \inprodS{ \nabla u}{  \nabla v}  + (2 \alpha^2 + 2 \abs{\nabla f}^2) \inprodS{u}{v} \\
		&\quad - \alpha (\inprodS{Du}{v} + \inprodS{u}{Dv}) \\
		& \quad  - \sum_j \bigl(\inprodS{\partial_ju}{\sigma_f \sigma_j v} + 2 (\partial_j f) \inprodS{u}{\partial_j v} \bigr)+ i \inprodS{ u}{\sigma_f D v}	    .
\end{align*}

We are interested in the integral over $\Omega$ of the above quantity with the weight $e^{-2 f}$.
If $u, v  \in \Dom{D^2}$, then we obtain using \eqref{eq : qform}
\begin{align*}
 \Scalprod{ e^{- 2  f}u}{D^2 v} 
	  &= \Scalprod {D e^{-2  f}u}{D v} \\
	  & = \Scalprod{ \nabla e^{- 2 f}u}{ \nabla v} + \int_\Bound  e^{- 2 f} \frac{\curv}{2}\inprodS{u}{ v} \\
	  &= \int_\Omega e^{- 2 f} \Bigl(\inprodS{ \nabla u}{\nabla v} - 2 \sum_j (\partial_j f) \inprodS{ u}{\partial_j v}\Bigr)
		  + \int_\Bound  e^{- 2 f} \frac{\curv}{2}\inprodS{u}{ v}.
\end{align*}
And by integration by parts 
\begin{align*}
\sum_j \int_\Omega e^{- 2 f} \inprodS{\partial_ju}{\sigma_f \sigma_j v} 
	  =&   - i \Scalprod{e^{- 2 f}u}{\sigma_f D v} + \Scalprod{(- \Delta f+ 2 \abs{\nabla f}^2 ) e^{- 2 f}u}{v} 
		\\
	&	+ \int_\Bound e^{- 2 f} \inprodS{u}{\sigma_f \dotprod{\sigma}{n} v}  .   
\end{align*}
Inserting these identities, we obtain
 \begin{align*}
 \sum_j \int_\Omega e^{- 2 f }& \inprodS{\widetilde \partial_j u}{\widetilde \partial_jv}
    = 
	    \Scalprod{e^{- 2 f} u}{D^2 v} + 2 \sum_j \Scalprod{e^{- 2 f} u}{  (\partial_j f ) \partial_j v}\\
	 &    +  \Scalprod{e^{- 2 f}(2 \alpha^2+ 2 \abs{\nabla f}^2)u}{v}  
-\alpha \Scalprod{e^{- 2 f} u}{ D v}   -\alpha \Scalprod{e^{- 2 f} D u}{v}\\
	 &  + 2 i  \Scalprod{e^{- 2 f} u}{\sigma_f D v}  + \Scalprod{e^{- 2 f} u}{(\Delta f -2\abs{\nabla f}^2) v} \\
	& - 2 \sum_j  \Scalprod{e^{- 2 f} u}{(\partial_j f )\partial_j v}
	          \\
	&+ \int_\Bound  e^{- 2 f} \bigl( - \frac{\curv}{2}\inprodS{u}{ v} - \inprodS{u}{\sigma_f \dotprod{\sigma}{n} v}\bigr).
 \end{align*}
 Note that the four terms containing first derivatives of $f$ cancel exactly.
Applying this identity with $v = u$ and $D u = \lambda u$, it reduces to
\begin{align*}
\sum_j  \int_\Omega e^{- 2 f} \inprodS{\widetilde \partial_j u}{\widetilde \partial_ju}
    = &
	    (\lambda^2 + 2 \alpha^2- 2\alpha \lambda) \norm{e^{-f}u}^2  
	      \\
	  &     + 2 i  \lambda \Scalprod{ u}{\sigma_f u} 
	    + \Scalprod{ e^{-2f} u }{\Delta fu}
		      \\
      &+ \int_\Bound  e^{- 2 f} \bigl( - \frac{\curv}{2}\inprodS{u}{u}-  \inprodS{u}{\sigma_f \dotprod{\sigma}{n} u}\bigr).    
 \end{align*}
Now we use the fact that the left hand side is real and nonnegative, and choose $\alpha = \lambda/2$ in order to minimize the coefficient of the first term.
We obtain the inequality
\begin{align*}
 0 \leq \frac{\lambda^2}{2} \norm{e^{-f}u}^2 +   \Scalprod{ e^{-2f} u }{(\Delta f)u}
 + \int_\Bound e^{-2  f}  \bigl( 
	     - \frac{\curv}{2}\inprodS{u}{u} - \Re  \inprodS{u}{\sigma_f \dotprod{\sigma}{n} u}\bigr) .
\end{align*}
Using anti-commutation relations,
$
\Re \inprodS{u}{\sigma_f \dotprod{\sigma}{n} u} = (\dotprod{n}{}\nabla f)  \inprodS{u}{u},
$
so we obtain 
\[
\frac{\lambda^2}{2} \norm{e^{- f} u}^2 \geq - \Scalprod{ e^{-2f} u }{(\Delta f)u}
+ \int_\Bound e^{-2 f}\inprodS{u}{u}( \frac{\curv}{2} + \dotprod{n}{}\nabla f).
\]
This suggests to take $f$ solving, for some $C\in\R$,
\begin{equation}
 \label{eq : neumann}
 \left\{ \begin{array}{ll}
         \Delta f = C & \text{ in } \Omega, \\
         \dotprod{n}{}\nabla f = - \kappa/2  & \text{ in } \Bound.
        \end{array} \right. 
\end{equation}
To see that such an $f$ exists, set $f_0(x) = C \abs{x}^2/4$.
By \cite[Theorem 3.40, p138]{bookFolland}, we can find $f_h$ satisfying
\[
\left\{ \begin{array}{ll}
         \Delta f_h = 0, & \text{ in } \Omega, \\
         \dotprod{n}{}\nabla f_h = - \kappa/2 - \dotprod{n}{}\nabla f_0,  & \text{ in } \Bound,
        \end{array} \right.
\]
provided $\int_\Bound ( - \kappa/2 - \dotprod{n}{}\nabla f_0)= 0$.
Since $\Omega$ is simply connected, $\int_\Bound \curv = 2 \pi$. On the other hand, $\int_\Bound \dotprod{n}{}\nabla f_0 = \int_\Omega \Delta f_0 = C \abs{\Omega}$.
So with the choice $C = -\pi / \abs{\Omega}$, $f_h+f_0$ satisfies \eqref{eq : neumann}.
The final result is
\[
\lambda^2 \geq  - 2 C = \frac{2 \pi}{\abs{\Omega}}.
\]

\end{proof}

\section{Application to the two-valley description of graphene}\label{sec : 2-valley}
We now apply our results to the description of electronic excitations in graphene as a
four-component spinor with
\[
H = \begin{pmatrix} T & 0 \\ 0 & T\end{pmatrix}.
\]
If the boundary conditions are local and uniform, the four-spinors should
fulfill 
\begin{align*}
  P_-(A)\psi :=\frac{1}{2}(1-A)\psi=0 \quad \text{on } \Bound.
\end{align*}
Here $A$ is a unitary matrix that belongs to a
four-parameter family, see \cite{AkhmerovBeenakker} for its explicit form.
For simplicity we will restrict our attention to the boundary
conditions  most commonly used in the physics literature, following the notations of \cite{AkhmerovBeenakker}.

\paragraph{ \bf Zigzag boundary conditions}
Zigzag boundary conditions arise from the tight-binding model when the honeycomb lattice is terminated in a direction perpendicular to the bonds, see Figure~\ref{fig: lattice}.
In this case, 
\[
A =  \begin{pmatrix}
     \sigma_3 & 0 \\ 0 & -\sigma_3
    \end{pmatrix}.
\]
Thus, these boundary conditions do not mix the two valleys.
We obtain two copies of $D_{\pi/2}$, which is not self-adjoint on
$H^1(\Omega,\C^2)$ and has zero as an eigenvalue of infinite
multiplicity \cite{schmidt1994}.

\begin{figure} 
\includegraphics[width = 0.35 \textwidth]{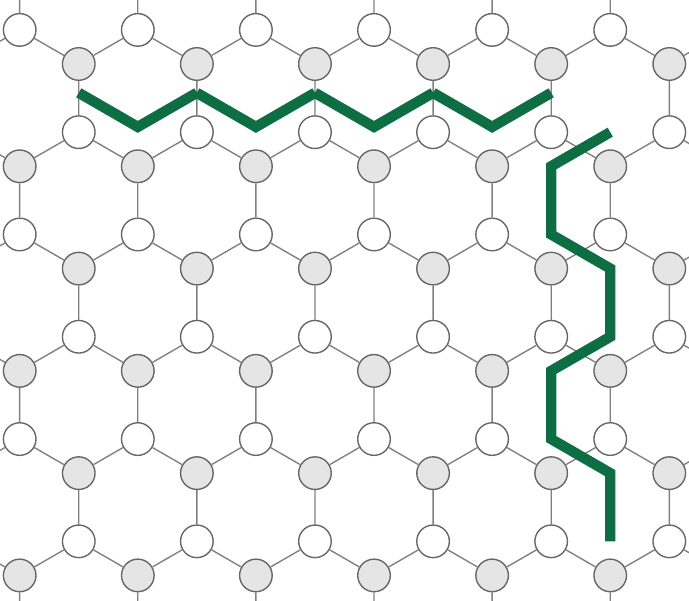}
\caption{A honeycomb lattice, where the gray and white dots represent carbon atoms on each of the two triangular sublattices. The thick lines indicate the zigzag boundary (above) and armchair boundary (right).}\label{fig: lattice}
\end{figure}

\paragraph{ \bf Infinite Mass boundary conditions}
These boundary conditions have been used as an effective model describing a graphene quantum dot or nanoribbon, when a detailed microscopical description of the boundary is lacking \cite{beneventano2014charge, CastroNetoetAl, PonomarenkoetAl}. 
In addition, infinite mass boundary conditions are obtained as a limiting case of $T$ acting on $L^2(\R^2)$ with a mass term $M(x) = M(1 - \chi_\Omega(x))$ when the mass tends to infinity. Physically, such a term is the continuum limit of a staggered potential (opposite signs on both sublattices) and several mechanisms to realize this in practice have been proposed \cite{GiovannettiKhomyakovetAl,SubramaniamLibischetAl, ZhouGweonetAl}. 

The matrix giving the boundary conditions here is 
\[
A= \begin{pmatrix}
 \dotprod{\sigma}{t} & 0 \\ 0 &- \dotprod{\sigma}{t} 
\end{pmatrix}.
\]
This does not mix the valleys and gives a block diagonal operator with $\D_0$ and $\D_\pi$ on the diagonal.
Therefore, the operator is self-adjoint with domain
$\Dom{D_0}\oplus\Dom{D_\pi}\subset  H^1(\Omega, \C^4)$ and the estimate of Theorem \ref{thm : lower_bound} holds.

\paragraph{ \bf Armchair boundary conditions}
Armchair boundary conditions also arise from the termination of a lattice, when the direction of the boundary is parallel to the bonds (see Figure~ \ref{fig: lattice}). It has been noted in some particular cases  that these boundary conditions give rise to a gap in the spectrum around zero (see for instance \cite{PhysRevB.73.235411,PhysRevB.88.125409, ZhengWangetAl}).
The boundary conditions are determined by 
\[
A =  \begin{pmatrix}
     0 & \nu^* \dotprod{\sigma}{t} \\ \nu \dotprod{\sigma}{t} & 0 
    \end{pmatrix}
\]
where $\abs{\nu}= 1$. We show how these boundary conditions can be brought into a block-diagonal form in order to apply our theorem.
Using the unitary transformation
\[
U_\nu = \begin{pmatrix}
     \nu 1_{\C^2} & 0 \\ 0 & 1_{\C^2} 
    \end{pmatrix}
\]
we can restrict our attention to the case $\nu = 1 $.
Consider the unitary transformation 
\[
U_p = \begin{pmatrix}
       1&0&0&0 \\ 0&0&0&1 \\ 0&0&1&0 \\ 0&1&0&0
      \end{pmatrix}
\]
corresponding  to a permutation of the second and fourth spinor
components.  This transforms the boundary conditions in 
\[
\widetilde A = U_p A U_p^\ast = \begin{pmatrix}
                                \dotprod{\sigma}{t} & 0 \\ 0 & \dotprod{\sigma}{t}
                               \end{pmatrix},
\]
and the Hamiltonian as
\[
\widetilde H = U_p H U_p^\ast  = \begin{pmatrix}
 0 & T \\ T & 0 
\end{pmatrix}.
\]
After this transformation, a four-spinor $\psi$ in the domain of $\widetilde H$ can be written as 
$\psi = \left(\begin{smallmatrix}            u_1 \\ u_2         \end{smallmatrix}\right)$ 
with $u_1, u_2 \in \Dom{D_0}$.
A short calculation shows that the same holds for the adjoint:
$\phi \in \Dom{\widetilde H^*}$ if and only if $\phi = \left(\begin{smallmatrix}            v_1 \\ v_2         \end{smallmatrix}\right)$ 
with $v_1, v_2 \in \Dom{D_0^*} = \Dom{D_0}$.
Thus, $H$ is self-adjoint on a domain included in $ H^1(\Omega, \C^4)$.
Furthermore,
\[
\norm{\widetilde{H}  \psi}^2 
= \norm{\begin{pmatrix}            D_0 u_2 \\ D_0 u_1         \end{pmatrix}}^2 
\geq \frac{2 \pi}{\abs \Omega} \norm{\psi}^2.
\]
In other words, the estimate of Theorem \ref{thm : lower_bound} holds
in this case as well.
\bigskip

\noindent
{\bf Acknowledgments.}
This work has  been
supported by the Iniciativa Cient\'ifica Milenio (Chile) through the
Millenium Nucleus RC–120002 ``F\'isica Matem\'atica'' . 
R.B.  has been supported by Fondecyt (Chile) Projects \# 112--0836 and
\#114--1155. S.F. acknowledges partial support from a Sapere Aude grant
from the Danish Councils for Independent Research, Grant number
DFF--4181-00221. E.S has been partially funded by Fondecyt (Chile)
project \# 114--1008. H. VDB. acknowledges support from Conicyt
(Chile) through CONICYT--PCHA/Doctorado Nacional/2014.  This work was
carried out while S.F. was invited professor at Pontificia Universidad
Cat\'olica de Chile.


\begin{thebibliography}{10}

\bibitem{AkhmerovBeenakker}
A.~R. Akhmerov and C.~W.~J. Beenakker, \emph{Boundary conditions for {D}irac
  fermions on a terminated honeycomb lattice}, Phys. Rev. B \textbf{77} (2008),
  085423.

\bibitem{AmmannBar}
B.~Ammann and C.~B{\"a}r, \emph{{D}irac eigenvalue estimates on surfaces},
  Math. Z. \textbf{240} (2002), no.~2, 423--449.

\bibitem{Bar1992}
C.~B{\"a}r, \emph{Lower eigenvalue estimates for {D}irac operators}, Math. Ann.
  \textbf{293} (1992), no.~1, 39--46.

\bibitem{beneventano2014charge}
C.~G. Beneventano, I.~Fialkovsky, E.~M. Santangelo, and D.~V. Vassilevich,
  \emph{Charge density and conductivity of disordered berry-mondragon graphene
  nanoribbons}, The European Physical Journal B \textbf{87} (2014), no.~3,
  1--9.


\bibitem{BenguriaFournaisetAl}
R.~Benguria, S.~Fournais, E.~Stockmeyer, and H.~Van Den~Bosch,
  \emph{Self--adjointness of Two-Dimensional Dirac operators in Domains}, 
 Annales Herin Poincar\'e (online first), DOI \, 10.1007/s00023-017-0554-5.


\bibitem{BerryMondragon}
M.~V. Berry and R.~J. Mondragon, \emph{Neutrino billiards: time-reversal
  symmetry-breaking without magnetic fields}, Proc. Roy. Soc. London Ser. A
  \textbf{412} (1987), no.~1842, 53--74.

\bibitem{bookBooss}
B.~Boo{\ss}-Bavnbek and K.~P. Wojciechowski, \emph{Elliptic boundary problems
  for {D}irac operators}, Mathematics: Theory \& Applications, Birkh\"auser
  Boston, Inc., Boston, MA, 1993.

\bibitem{PhysRevB.73.235411}
L.~Brey and H.~A. Fertig, \emph{Electronic states of graphene nanoribbons
  studied with the dirac equation}, Phys. Rev. B \textbf{73} (2006), 235411.

\bibitem{CastroNetoetAl}
A.~H. Castro~Neto, F.~Guinea, N.~M.~R. Peres, K.~S. Novoselov, and A.~K. Geim,
  \emph{The electronic properties of graphene}, Rev. Mod. Phys. \textbf{81}
  (2009), 109--162.

\bibitem{Fefferman}
C.L. Fefferman and M.I. Weinstein, \emph{Honeycomb lattice potentials and
  {D}irac points}, J. Amer. Math. Soc. \textbf{25} (2012), no.~4, 1169--1220.

\bibitem{bookFolland}
G.~B. Folland, \emph{Introduction to partial differential equations}, second
  ed., Princeton University Press, Princeton, NJ, 1995.

\bibitem{freitassiegl2014}
P.~Freitas and P.~Siegl, \emph{Spectra of graphene nanoribbons with armchair
  and zigzag boundary conditions}, Rev. Math. Phys. \textbf{26} (2014), no.~10,
  1450018, 32.

\bibitem{Friedrich1980}
Th. Friedrich, \emph{Der erste {E}igenwert des {D}irac-{O}perators einer
  kompakten, {R}iemannschen {M}annigfaltigkeit nichtnegativer
  {S}kalarkr\"ummung}, Math. Nachr. \textbf{97} (1980), 117--146.

\bibitem{GiovannettiKhomyakovetAl}
G.~Giovannetti, P.~A. Khomyakov, G.~Brocks, P.~J. Kelly, and J.~van~den Brink,
  \emph{Substrate-induced band gap in graphene on hexagonal boron nitride:
  \textit{Ab initio} density functional calculations}, Phys. Rev. B \textbf{76}
  (2007), 073103.

\bibitem{HijaziMontielZhang}
O.~Hijazi, S.~Montiel, and X.~Zhang, \emph{Eigenvalues of the {D}irac operator
  on manifolds with boundary}, Comm. Math. Phys. \textbf{221} (2001), no.~2,
  255--265.

\bibitem{mccannfalko04}
E.~McCann and V.~I. Fal’ko, \emph{Symmetry of boundary conditions of the
  dirac equation for electrons in carbon nanotubes}, Journal of Physics:
  Condensed Matter \textbf{16} (2004), no.~13, 2371.

\bibitem{PhysRevB.88.125409}
A.~Orlof, J.~Ruseckas, and I.~V. Zozoulenko, \emph{Effect of zigzag and
  armchair edges on the electronic transport in single-layer and bilayer
  graphene nanoribbons with defects}, Phys. Rev. B \textbf{88} (2013), 125409.

\bibitem{PonomarenkoetAl}
L.~A. Ponomarenko, F.~Schedin, M.~I. Katsnelson, R.~Yang, E.~W. Hill, K.~S.
  Novoselov, and A.~K. Geim, \emph{Chaotic {D}irac billiard in graphene quantum
  dots}, Science \textbf{320} (2008), no.~5874, 356--358.


\bibitem{Raulot06}
S.~Raulot, \emph{The Hijazi inequality on manifolds with boundary},  J. Geom. Phys. \textbf{56} (2006), 2189--2202.


\bibitem{RitterLyding}
K.~A. Ritter and Joseph~W. Lyding, \emph{The influence of edge structure on the
  electronic properties of graphene quantum dots and nanoribbons}, Nat. Mater.
  (2009), 235.

\bibitem{schmidt1994}
K.~M. Schmidt, \emph{A remark on boundary value problems for the {D}irac
  operator}, Quart. J. Math. Oxford Ser. (2) \textbf{46} (1995), no.~184,
  509--516.

\bibitem{stockmeyerwugalter2016}
E.~Stockmeyer and S.~Vugalter, \emph{Infinite mass boundary conditions for
  {D}irac operators}, Preprint (2016), arXiv:1603.09657.

\bibitem{SubramaniamLibischetAl}
D.~Subramaniam, F.~Libisch, Y.~Li, C.~Pauly, V.~Geringer, R.~Reiter,
  T.~Mashoff, M.~Liebmann, J.~Burgd\"orfer, C.~Busse, T.~Michely,
  R.~Mazzarello, M.~Pratzer, and M.~Morgenstern, \emph{Wave-function mapping of
  graphene quantum dots with soft confinement}, Phys. Rev. Lett. \textbf{108}
  (2012), 046801.

\bibitem{wallace}
P.R. Wallace, \emph{The band theory of graphite}, Phys. Rev. \textbf{71}
  (1947), no.~9, 622.

\bibitem{WurmRycerzetAl}
J.~Wurm, A.~Rycerz, \ifmmode \dot{I}\else \.{I}\fi{} \ifmmode
  \mbox{\c{c}}\else~\c{c}\fi{} Adagideli, M.~Wimmer, K.~Richter, and H.~U.
  Baranger, \emph{Symmetry classes in graphene quantum dots: Universal spectral
  statistics, weak localization, and conductance fluctuations}, Phys. Rev.
  Lett. \textbf{102} (2009), 056806.

\bibitem{ZhengWangetAl}
H.~Zheng, Z.~F. Wang, T.~Luo, Q.~W. Shi, and J.~Chen, \emph{Analytical study of
  electronic structure in armchair graphene nanoribbons}, Phys. Rev. B
  \textbf{75} (2007), 165414.

\bibitem{ZhouGweonetAl}
S.~Y. Zhou, G.-H. Gweon, A.~V. Fedorov, P.~N. First, W.~A. de~Heer, D.-H. Lee,
  F.~Guinea, A.~H. Castro~Neto, and A.~Lanzara, \emph{Substrate-induced bandgap
  opening in epitaxial graphene}, Nat. Mater. (2007), 770.

\end{thebibliography}
\end{document}